\documentclass[runningheads]{llncs}
\usepackage{algpseudocode}
\usepackage{amsmath}
\usepackage{amssymb}
\usepackage{caption} 
\usepackage{subcaption}
\usepackage{graphicx}
\usepackage{algorithm}
\usepackage{xcolor}
\usepackage{cleveref}

\usepackage{lineno}
\usepackage[]{todonotes}

\usepackage{tabularx}

\newcommand{\TReclong}{\textsc{Temporal Recolouring}}
\newcommand{\TRec}{\textsc{TRec}}

\newcommand{\problemdef}[3]{
	\begin{center}
	\begin{minipage}{0.95\textwidth}
		\noindent
		#1
		\vspace{5pt}\\
		\setlength{\tabcolsep}{3pt}
		\begin{tabularx}{\textwidth}{@{}lX@{}}
			\textbf{Input:}     & #2 \\
			\textbf{Task:}  & #3
		\end{tabularx}
	\end{minipage}
	\end{center}
}

\newcommand{\decproblemdef}[3]{
	\begin{center}
	\begin{minipage}{0.95\textwidth}
		\noindent
		#1
		\vspace{5pt}\\
		\setlength{\tabcolsep}{3pt}
		\begin{tabularx}{\textwidth}{@{}lX@{}}
			\textbf{Input:}     & #2 \\
			\textbf{Question:}  & #3
		\end{tabularx}
	\end{minipage}
	\end{center}
}

\DeclareMathOperator{\Components}{Components}

\DeclareMathOperator{\Cost}{Cost}

\usepackage[T1]{fontenc}
\usepackage{graphicx}
\newtheorem{observation}{Observation}

\begin{document}
\bibliographystyle{plainurl}
\title{Maintaining Bipartite Colourings on Temporal Graphs on a Budget}
%
%
\author{Duncan Adamson\inst{1}\orcidID{0000-0003-3343-2435} \and
 George B. Mertzios\inst{2}\orcidID{0000-0001-7182-585X} \and
 Paul G. Spirakis\inst{3}\orcidID{0000-0001-5396-3749}}
%




\authorrunning{Duncan Adamson, George B. Mertzios, and Paul G. Spirakis} 


%
\institute{School of Computer Science, University of St Andrews, UK
\email{duncan.adamson@st-andrews.ac.uk}\and
Department of Computer Science, University of Durham, UK\\
\email{george.mertzios@durham.ac.uk} \and
Department of Computer Science, University of Liverpool, UK\\
\email{p.spirakis@liverpool.ac.uk}
}
\maketitle              
\begin{abstract}
Graph colouring is a fundamental problem for networks, serving as a tool for avoiding conflicts via symmetry breaking, for example, avoiding multiple computer processes simultaneously updating the same resource. This paper considers a generalisation of this problem to \emph{temporal graphs}, i.e., to graphs whose structure changes according to an ordered sequence of edge sets. 
In the simultaneous resource updating problem on temporal graphs, the resources which can be accessed will change, however, the necessity of symmetry breaking to avoid conflicts remains.

In this paper, we focus on the problem of \emph{maintaining proper colourings} on temporal graphs in general, with a particular focus on bipartite colourings. Our aim is to minimise the total number of times that the vertices change colour, or, in the form of a decision problem, whether we can maintain a proper colouring by allowing not more colour changes than some given \emph{budget}. On the negative side, we show that, despite bipartite colouring being easy on static graphs, the problem of maintaining such a colouring on graphs that are bipartite in each snapshot is NP-Hard to even approximate within \emph{any} constant factor unless the Unique Games Conjecture fails. On the positive side, we provide an exact algorithm for a temporal graph with $n$ vertices, a lifetime $T$ and at most $k$ components in any given snapshot in $O(T \vert E \vert 2^{k} + n T 2^{2k})$ time, and an $O\left(\sqrt{\log(nT)}\right)$-factor approximation algorithm running in $\tilde{O}((nT)^3)$ time.

Our results contribute to the structural complexity of networks that change with time with respect to a fundamental computational problem.

\end{abstract}
\newpage
\section{Introduction}

Graph colouring is one of the most fundamental questions in graph theory. Informally, this problem asks for an assignment of colours from a palette, normally represented by a set of integers $1, 2, \dots, C$, such that no two adjacent vertices share a colour. This has many applications, particularly in networks where coordination is needed between various agents to avoid conflicts, for example, avoiding system resources being updated simultaneously. Given the broad applicability of this problem, there is a large body of work for both centralised \cite{formanowicz2012survey,matula1972graph} and distributed \cite{fuchs2025distributed,ghaffari2022deterministic,HalldorssonKNT22} computing. See \cite{lewis2021guide} for an overview of several applications.

In this paper, we introduce and study a new generalisation of this problem to \emph{temporal graphs}, graphs formed by a sequence of \emph{snapshots}, with changing edgesets over a constant set of vertices. In our variant, the goal to to construct a sequence of colourings such that each snapshot is properly coloured, while minimising the number of changes between snapshots. This allows us to capture the problems of avoiding conflicts between neighbouring vertices, while also incorporating the changing nature of temporal graphs. To date, there has been a reasonable body of work on colouring temporal graphs \cite{adamson2025harmonious,barba2017dynamic,ibiapina2025color,marino2022coloring,mertzios2021sliding,yu2013algorithms}, with the primary, but by no means exclusive, focus being on \emph{sliding window colourings}, colourings where each edge is properly coloured at least once within each window.

In this paper, we focus primarily on \emph{bipartite graphs}, graphs which can be coloured with a palette of size two, traditionally $\{0, 1\}$. Our primary reason for this focus is that, even for a palette of size three, this problem is trivially NP-hard, from the general problem of 3-colouring graphs. On the other hand, finding a bipartite colouring can be done in time linear relative to the number of edges, making this setting far more interesting from a technical standpoint.

We mention some of the other results on generalisations of the colouring problem to temporal graphs. First is the work by Yu et al. \cite{yu2013algorithms}, who study a similar definition to ours, with the objective of minimising the sum $C + \alpha A$ where $C$ is the palette size, $A$ is total number of changes, and $\alpha$ is some parameter chosen by the user. The authors study six colouring algorithms, providing both an experimental comparison and explicit classes in which the algorithms perform best. We note that, unlike this work, they do not provide constraints on the size of the palette, and while their formulation remains NP-hard, the reduction follows explicitly from the hardness of the colouring problem in general, without any reference to bipartite graphs. More recent is the work by Mertzios et al. \cite{mertzios2021sliding} on sliding window colourings, who showed that the problem of finding such a problem is NP-hard even in several restricted classes, while also providing strong FPT results for finding such colourings in general. This work was built upon by Marino and Silva \cite{marino2022coloring}, who study parameters under which such a colouring can be found in polynomial time, specifically relating to how frequent, and for how long, each edge is active.

\paragraph*{Our Contribution.}

We show that this problem restricted to bipartite graphs occupies an interesting position in terms of complexity. Determining whether there exists a colouring that does not require any vertex to be recoloured in any snapshot can be done in linear time, however, determining the minimum number of changes needed for any non-trivial instance is NP-hard. On the positive side, for a temporal graph with $n$ vertices and a lifetime $T$, we present an algorithm that is a fixed-parameter tractable in the maximum number $k$ of connected components in each snapshot, running in $O(T \vert E \vert 2^{k}  + n T 2^{2k})$ time, and an approximation algorithm with a factor of $O\left(\sqrt{\log(nT)}\right)$ running in $\tilde{O}((n T)^3)$ time. To complete the picture, we prove that, for every $C\geq 3$, it is NP-hard even to determine whether there exists a $C$-colouring with recolouring cost zero, i.e., a $C$-colouring that does not require any vertex to be recoloured in any snapshot.

\section{Preliminaries}

We denote by $[i]$ the ordered sequence of integers $1, 2, \dots, i$, and by $[i, j]$ the sequence $i, i + 1, \dots, j$, for any $i, j \in \mathbb{Z}$, $i \leq j$. We define a \emph{temporal graph} by an ordered sequence of \emph{snapshots}, each of which is a graph over a common set of vertices. Formally, let $\mathcal{G}$ be a temporal graph, then $\mathcal{G} = G_1, G_2, \dots, G_T$ where $G_t = (V, E_t)$. We call the number of snapshots, by convention denoted $T$, the \emph{lifetime} of the temporal graph, and $G_t$ the $t^{\text{th}}$ \emph{snapshot} of the temporal graph.
Where confusion may arise, we refer to non-temporal graphs as \emph{static graphs}. The \emph{underlying graph} of a temporal graph, $\mathcal{G}$, denoted $U(\mathcal{G})$, is the static graph formed by the union of all edge sets in the temporal graph, i.e., $U(\mathcal{G}) = \left(V, \bigcup_{t \in [T]} E_t\right)$.

Given a vertex $v \in V$ in a temporal graph $\mathcal{G}$, we denote by $N_t(v)$ the set of \emph{neighbours} of $v$ in the snapshot $G_t = (V, E_t)$, formally $N_t(v) = \{u \in V \mid (v, u) \in E_t\}$. The \emph{degree} of a vertex $v$ in the snapshot $G_t$, denoted $\deg_t(v)$, is the number of neighbours of $v$ in $G_t$, formally $\deg_t(v) = \vert N_t(v) \vert$. We ommit the subscript denoting snapshot for both the set of neighbours and degree of a vertex in a static graph, thus $N(v)$ denotes the set of neighbours, and $\deg(v) = \vert N(v) \vert$ the degree. Given multiple static graphs, we denote by $N_G(v)$ (resp., $\deg_G(v)$) the set of neighbours (resp. degree) of $v$ in the graph $G$. A \emph{component} in a static graph $G = (V, E) $ is a subgraph $K = (V_K, E_K)$ such that $\forall v \in V_K$, $N(v) \subseteq V_K$ and $E_K = E \cap  (V_K \times V_K)$.

We define a colouring of a graph as a function $\psi: V \mapsto [0, C - 1]$ for some palette size $C$. We call a colouring over a palette of size $C$ a \emph{$C$-colouring}. A $2$-colouring, as studied in this paper, is called \emph{bipartite}, assigning to each vertex a colour in the set $\{0, 1\}$. A colouring of the (static) graph $G = (V, E)$ is \emph{proper} if, given any edge $(v, u) \in E$, we have $\psi(v) \neq \psi(u)$. An edge $(v, u)$ is called \emph{monochrome} if $\psi(v) = \psi(u)$. For the remainder of this paper, we assume by default that any given colouring is proper unless explicitly stated to be otherwise. We denote by $\phi_C(G)$ the set of all proper $C$-colourings of a given static graph~$G$.

We define a \emph{temporal sequence $C$-colouring} of a temporal graph of lifetime~$T$ as a sequence of $T$ proper $C$-colourings, $\Psi = \psi_1, \psi_2, \dots, \psi_T$. We denote by $\Phi_C(\mathcal{G})$ the set of all temporal sequence $C$-colourings of the temporal graph $\mathcal{G}$. The \emph{cost} of a temporal sequence colouring, denoted $\Cost(\Psi)$, is the total number of colour changes of each vertex over the lifetime of the graph. Formally,

$$\Cost(\Psi) = \sum_{v \in V} \sum_{t \in [T - 1]} \begin{cases}
    0 & \psi_t(v) = \psi_{t + 1}(v)\\
    1 & \psi_t(v) \neq \psi_{t + 1}(v)
\end{cases}.$$

A \emph{minimum cost temporal sequence $C$-colouring} of a temporal graph $\mathcal{G}$ is a temporal sequence $C$-colouring $\Psi \in \Phi_C(\mathcal{G})$ such that $\Cost(\Psi) \leq \Cost(\Psi')$, $\forall \Psi' \in \Phi_C(\mathcal{G})$.
The primary problem we study in this paper is \TReclong. As a decision problem, we ask whether there exists a temporal sequence $C$-colouring of a given temporal graph $\mathcal{G}$ with a cost of at most some given budget $B \in \mathbb{N}$. In the optimisation version of the problem, we are not given $B$ along with the input, but we rather aim at computing the cost of a minimum cost temporal sequence colouring.


\decproblemdef{\TReclong \ (\TRec)}
{A temporal graph $\mathcal{G}$, a palette of size $C\in\mathbb{N}$, and budget $B\in\mathbb{N}$.}
{Does there exist a temporal sequence $C$-colouring $\Psi : V \times [T] \mapsto [C]$ such that $\Cost(\Psi) \leq B$?}

Furthermore, whenever the size $C$ of the palette in the problem \TRec\ is a constant (and not part of the input), we refer to the problem as $C$-\TRec.


\subsection{Non-Bipartite Temporal Sequence Colourings}

Before presenting our main results, we make a small number of observations on some general properties of temporal sequence colourings for non-bipartite graphs. First, we justify our focus on bipartite colourings with the following:
\begin{observation}
    \TRec\ is NP-Hard for any $C \geq 3$.
\end{observation}

Note that this follows from the hardness of finding a 3-colouring in static graphs. We extend to the more complex setting of $\Delta + 1$ colourings, where $\Delta$ represents the maximum degree of any vertex in any snapshot, i.e. $\Delta = \max_{v \in V} \max_{t \in [T]} \deg_t(v)$. Unlike a 3-colouring a $(\Delta + 1)$-colouring can be found in $O(\vert E \vert)$ time for any given static graph $G = (V, E)$, and thus admits a less trivial proof of hardness.

\begin{proposition}
    \label{prop:deg_plus_one}
    For every $C\geq 3$, it is NP-hard to determine whether the minimum cost of a temporal sequence $C$-colouring of a temporal graph $\mathcal{G}$ is zero, even when every vertex has degree at most 1 in every snapshot.
\end{proposition}

\begin{proof}
    Let $G = (V, E)$ be a static graph, and let $m=|E|$ be the number of its edges and $C\geq 3$ be arbitrary. 
    We construct a temporal graph $\mathcal{G} = G_1, G_2, \dots, G_m$, 
    where every snapshot has exactly one of the edges of $G$. 
    Now, observe that no vertex can have a degree greater than 1 in any given snapshot. 
    Furthermore, note that there exists a temporal sequence $C$-colouring in this graph with a cost of $0$ if and only if there exists a $C$-colouring of $G$. Therefore, since $C\geq 3$, this problem is NP-hard.
    
\vspace{-12.1pt}
\qed\end{proof}

We note that Proposition \ref{prop:deg_plus_one} does not directly apply to bipartite graphs, as we can determine if a zero-cost temporal 2-colouring exists in linear time by determining if the underlying graph is bipartite. Thus, our results in the following sections are directly motivated in closing this bound.

\section{Hardness Results for Bipartite Colourings}

Our first major result is in showing that the \TRec\ is NP-hard and hard to approximate (subject to specific plausible computational complexity assumptions), even when every snapshot only consists of a disjoint union of paths (see~\Cref{thm:alw-bip_hardness}). For the proof of \Cref{thm:alw-bip_hardness}
we present a reduction from the problem \textsc{MinUnCut}, which is the dual problem of \textsc{MaxCut}:

\newcommand{\MUC}{\textsc{MinUnCut}}
\problemdef{\MUC}
{A (static) graph $G=(V,E)$.}
{Partition the vertex set $V$ of $H$ into two colour-classes such that the number of monochromatic egdes in $E$ is as small as possible.}


It is known that \MUC~is NP-hard and that it is known that there does not exist any polynomial-time constant approximation algorithm unless the Unique Games Conjecture fails~\cite{Khot02a}.

\begin{theorem}
    \label{thm:alw-bip_hardness}
    \TRec\ is NP-hard, even when $C=2$, $\mathcal{G}$ consists of a disjoint union of paths at every time step, and the lifetime of $\mathcal{G}$ is 2. Further, there is no polynomial-time constant approximation algorithm unless the Unique Games Conjecture fails.
\end{theorem}

\paragraph*{Construction.}
From an input instance $H$ of \MUC with $n$ vertices $v_1, \ldots, v_n$ and $m$ edges, we build the temporal graph $\mathcal{G} = G_{1}, G_{2}$, as follows. For every $i=1,\ldots,n$ denote by $d_i$ the number of neighbours $\deg_H(v_i)$ of $v_i$ in the graph $H$. 
At a high level, we represent each vertex by a set of $2d_i - 1$ vertices, connected in $G_{1}$ to form a \emph{vertex gadget}. In $G_{2}$, we replace the vertex gadgets with \emph{edge gadgets}. We represent the colours of the vertices of $H$ by the initial colouring of the vertex gadgets. 
We now provide the explicit construction.

The vertex gadget of each vertex $v_i$ is a path with $2d_i - 1$ vertices. The vertex set of $\mathcal{G}$ consists of these $\sum_{i=1}^{n}(2d_i - 1)=4m-n$ vertices. The vertices of the vertex gadget of $v_i$ are denoted $a_{i}^1, b_{i}^1, \ldots, a_{i}^{d_i-1}, b_{i}^{d_i-1}, a_{i}^{d_i}$, as shown in Figure~\ref{alw-bip-fig}, where the edges of this path are $\{(a_{i}^j, b_i^j), (a_i^j, b_{i}^{j + 1}) \vert j =1,\ldots,d_i-1\}$. Note that for every 2-colouring of this (static) vertex gadget, all vertices $a_{i}^1, \ldots, a_{i}^{d_i}$ are coloured with one colour and all vertices $b_{i}^1, \ldots, b_{i}^{d_i-1}$ are coloured with the other colour. Our first snapshot, $G_{1}$, corresponds exactly to the set of variable gadgets.


\begin{figure}[ht!]
        \centering
        \includegraphics[width=\linewidth]{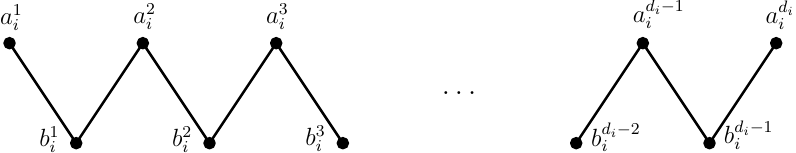}
        \caption{\vspace{0,1cm}}
    
    \caption{The variable gadget for vertex $v_i$.\label{alw-bip-fig}}
\end{figure}

In the second snapshot $G_{2}$, we remove all edges from all variable gadgets and, for each edge of the input graph $H$, we add one edge to $G_2$, as follows. For every vertex $v_i$ in the input graph $H$, we arbitrarily enumerate its edges. Let now $(v_i, v_j) \in E$ be an edge of $H$, which is the $p^{\text{th}}$ edge of vertex $v_i$ and the $q^{\text{th}}$ edge of vertex $v_j$, where $1\leq p \leq d_i$ and $1\leq q \leq d_j$. Then, we add to $G_2$ the edge $(a_i^p, a_j^q)$. That is, the second snapshot $G_2$ is a matching of $m$ edges.
This completes the construction of the temporal graph $\mathcal{G}=(G_1,G_2)$.

\begin{lemma}
    \label{lem:SAT_gadget_is_correct}
    Given an instance $H=(V,E)$ of \MUC~with $m$ clauses, the temporal graph $\mathcal{G}$ admits a $2$-colouring with a cost at most $k$ if and only if there exists a partition of $V$ into two colour classes such that there are at most $k$ monochromatic edges in $E$.
\end{lemma}

\begin{proof}
($\Leftarrow$) Let $(A,B)$ be a partition of $V$ into two colour classes such that there are at most $k$ monochromatic edges in the colouring induced by $(A, B)$ on $H$, i.e, $\vert E \cap (A \times A)\vert + \vert E \cap (B \times B)\vert \leq k$. In the first snapshot $G_{1}$, for every $i=1,\ldots,n$, we assign colours to the vertices of the variable gadget as follows: if $v_i\in A$, we assign colour 1 to vertices $a_{i}^1, \ldots, a_{i}^{d_i}$ and colour 0 to vertices $b_{i}^1, \ldots, b_{i}^{d_i-1}$. 
If $v_i\in B$, we assign colour 0 to vertices $a_{i}^1, \ldots, a_{i}^{d_i}$ and colour 1 to vertices $b_{i}^1, \ldots, b_{i}^{d_i-1}$. 

Consider now an arbitrary edge $(v_i, v_j)$ of $H$, which is the $p^{\text{th}}$ edge of $v_i$ and the $q^{\text{th}}$ edge of $v_j$. If $v_i$ and $v_j$ belong to different colour classes in $H$ (i.e., if either $v_i\in A$ and $v_j\in B$, or $v_i\in B$ and $v_j\in A$), then the edge $(v_i, v_j)$ is not monochromatic in $H$, and the edge $(a_i^p, a_j^q)$ is also properly coloured in $G_2$. 
Otherwise, if both $v_i$ and $v_j$ belong to the same colour class of $H$ (i.e., if either $v_i,v_j \in A$ or $v_i,v_j \in B$) then the edge $(v_i, v_j)$ is monochromatic in $H$; in this case, we flip the colour of vertex $a_i^p$ in $G_2$, such that the edge $(a_i^p, a_j^q)$ becomes properly coloured in $G_2$. 

Therefore, if with the vertex partition $(A,B)$ we have $t\leq k$ monochromatic edges in $H$, then in the transition from $G_1$ to $G_2$ we recolour exactly $t$ vertices. Therefore, we constructed a proper temporal sequence colouring of $\mathcal{G}$ with a palette $C=\{0,1\}$ and cost at most $k$.

($\Rightarrow$) Let $\Psi$ be a minimum cost temporal sequence colouring of $\mathcal{G}$ with a palette $C=\{0,1\}$ and cost at most $k$. By the construction of $\mathcal{G}$, for every $i=1,\ldots,n$, all vertices $a_{i}^1, \ldots, a_{i}^{d_i}$ are coloured with the same colour in $G_{1}$. 
We define from the temporal sequence colouring $\Psi$ a vertex partition $(A,B)$ of the input graph $H$ as follows. For every $i=1,\ldots,n$, we add $v_i$ to $A$ (resp.~to $B$) if the vertices $a_{i,1}^1, \ldots, a_{i,1}^{d_i}$ are coloured~1 (resp.~0) in~$G_{1}$.

Consider an arbitrary edge $a_i^p a_j^q$ of the second snapshot $G_2$. By the construction of $\mathcal{G}$, this edge corresponds to the edge $(v_i, v_j)$ of $H$ which is the $p^{\text{th}}$ edge of $v_i$ and the $q^{\text{th}}$ edge of $v_j$. 
First, suppose that, in the transition from $G_1$ to $G_2$ in $\mathcal{G}$, the colouring $\Psi$ recolours neither $a_i^p$ nor $a_j^q$. Then, by the construction of the vertex partition $(A,B)$ of $H$, the edge $(v_i, v_j)$ is not monochromatic in the partition $(A,B)$. That is, either $v_i\in A$ and $v_j\in B$, or $v_i\in B$ and $v_j\in A$. 
Now suppose that, in the transition from $G_1$ to $G_2$ in $\mathcal{G}$, the colouring $\Psi$ recolours one of the vertices $a_i^p$ or $a_j^q$. Then, by the construction of the vertex partition $(A,B)$ of $H$, the edge $(v_i, v_j)$ is monochromatic in the partition $(A,B)$. That is, either $v_i,v_j\in A$ or $v_i,v_j\in B$. 

Therefore, if the minimum cost temporal sequence colouring $\Psi$ of $\mathcal{G}$ has cost $t\leq k$, then in the vertex partition $(A,B)$ of $H$ we have exactly $t\leq k$ monochromatic edges. This completes the proof.

\vspace{-12.1pt}
\qed\end{proof}

\section{$O\left(\sqrt{\log(nT)}\right)$-Approximation}

We now present our first positive algorithmic contribution in the form of an algorithm to find an $O\left(\sqrt{\log(nT)}\right)$ approximation of the minimum budget needed to maintain a bipartite colouring. Recall that we can determine if there exists a zero-cost colouring in $O\left(\sum_{t \in [T]} \vert E_t \vert\right)$ time by checking if the underlying graph is bipartite. Thus, our algorithm is a ``true'' approximation in that it returns a non-optimal budget if and only if there is no solution where the budget is zero.\looseness=-1

Our approximation itself stems from the $O\left(\sqrt{\log(n)}\right)$-approximation of the \MUC~problem due to Agarwal et al. \cite{agarwal2005log}, computable in $\tilde{O}(n^{3})$ time due to Arora et al \cite{arora2007combinatorial}, where $\tilde{O}(n^{3}) = O(n^{3} \log^c n)$ for some constant $c$.\looseness=-1


Our approach is to formulate the problem of finding a minimum budget temporal sequence $2$-colouring of a temporal graph $\mathcal{G}$ as a \MUC~instance, allowing this approximation algorithm to be applied directly.

\paragraph*{Formulation.}

The high-level idea behind our formulation is to construct an auxiliary static graph $\alpha(G) = (V_{\alpha}, E_{\alpha})$ such that the number of monochromatic edges in a partition of $V_\alpha$ corresponds to at most twice the cost of some temporal sequence colouring. The vertex set $V_\alpha$ contains a new vertex for each vertex in $V = \{v_1, v_2, \dots, v_n\}$ and snapshot $G_1, G_2, \dots, G_T$, formally $V_\alpha = \{v_{i, t} \mid i \in [n], t \in [T]\}$, where $v_{i, t}$ represents the vertex $v_i$ in snapshot $G_t$. We define our edge set $E_\alpha$ by two subsets, $C$ and $S$. The set $C$ connects vertices between snapshots, which we use to represent the cost of changing the colour of a given vertex. Formally, $C = \{(v_{i, t}, v_{i,t + 1}) \mid \forall v_i \in V, t \in [T - 1]\}$.

The set $S$ represents the components in the temporal graph within the new static graph. Given a component $K$ containing the set of vertices $V_K = \{ v_1, v_2, \dots, v_m \}$ in $G_t$ and 2-colouring $\psi$ of $K$, we add the edge $(v_{i, t}, v_{j, t})$ to $S$ for every pair $v_i, v_j \in V_K$ where $\psi(v_i) \neq \psi(v_j)$. This way, if $(v_{i, t}, v_{j, t}) \in S$, then $\psi_t$ is a proper colouring of $G_t$ only if $\psi_t(v_i) \neq \psi_t(v_j)$.
%

\begin{lemma}
    \label{lem:colouring_to_auxilary_graph}
    Let $\Psi$ be a temporal sequence $2$-colouring of the temporal graph $\mathcal{G} = (G_1, G_2, \dots, G_T)$. Then, there exists a $2$-colouring of $\alpha(\mathcal{G})$ with $\Cost(\Psi)$ monochromatic edges.
\end{lemma}

\begin{proof}
    We define our 2-colouring on $\alpha(\mathcal{G})$ by a pair of sets $A$ and $B$, constructed as follows. Given a vertex $v_i$ and snapshot $G_t$, we assign $v_{i, t}$ to the set $A$ if $(\psi_t(v_i) + t) \bmod 2 \equiv 0$, and to set $B$ otherwise. We use this construction so that if $\psi_{t}(v_i) = \psi_{t + 1}(v_{i})$, then either $v_{i, t} \in A$ and $v_{i, t + 1} \in B$, or $v_{i, t} \in B$ and $v_{i, t + 1} \in A$. Otherwise, if the colour of $v_i$ changes between snapshots, we have $v_{i, t}, v_{i, t + 1} \in A$ or $v_{i, t}, v_{i, t + 1} \in B$.
    
    Observe that, as $\psi_t$ is a proper colouring of $G_t$, we have that, for any edge $(v_i, v_j) \in E_t$, $\psi_t(v_i) \neq \psi_t(v_j)$, thus either $v_{i, t} \in A$ and $v_{j, t} \in B$, or $v_{i, t} \in B$ and $v_{j, t} \in A$. Similarly, if $(v_{i, t}, v_{j, t}) \in E_{\alpha}$ and $(v_i, v_j) \notin E_t$, then $v_i$ and $v_j$ are in the same component in $G_t$, with the condition that they must belong to different colour classes of any bipartite colouring of the component. Hence, $(v_i, v_j)$ can not form a monochromatic edge in $V_{\alpha}$ without contradicting the assumption that $\psi_t$ is a proper colouring of $G_t$.
    
    Therefore, any monochromatic edge in the colouring on $\alpha(\mathcal{G})$ induced by the partition $A$ and $B$ must be of the form $(v_{i, t}, v_{i, t + 1})$, for some $t \in [T - 1]$. Further, by construction, $(v_{i, t}, v_{i, t + 1})$ will be a monochromatic edge in $\alpha(\mathcal{G})$ if and only if $\psi_t(v_i) \neq \psi_{t + 1}(v_i)$. Thus, the total number of monochromatic edges in this partition is exactly equal to the cost of $\Psi$.

\vspace{-12.1pt}
\qed\end{proof}

\begin{lemma}
    \label{lem:auxilary_graph_to_bipartite}
    Let $\mathcal{G} = (G_1, G_2, \dots, G_T)$ be a temporal graph and let $A, B$ be a partition of $V_\alpha$ inducing at most $C$ monochromatic edges in $\alpha(\mathcal{G})$. Then, there exists a temporal sequence colouring $\Psi$ of $\mathcal{G}$ with a cost of at most $2C$.
\end{lemma}

\begin{proof}
    We initially set $\psi_t(v_i)$ to $0$ if either $v_{i, t} \in A$ and $t \bmod 2 \equiv 0$, or $v_{i, t} \in B$ and $t \bmod 2 \equiv 1$. Otherwise, we set $\psi_t(v_i)$ to $1$. Note this matches the construction used in Lemma \ref{lem:colouring_to_auxilary_graph} to create a partition from the colouring, with the same property that a vertex changes colours between snapshots if and only if there is a corresponding monochromatic edge in the partition of the auxiliary graph. Therefore, given a monochromatic edge in the partition of $V_{\alpha}$ between some pair of vertices of the form $(v_{i, t}, v_{i, t + 1})$, $v_i$ changes colour in the corresponding colouring, and thus the colouring has an associated cost of $1$ for each such edge.

    Now, consider some component $K$ with the vertex set $V_K = \{v_1, v_2, \dots, v_m\}$ in snapshot $G_t$ such that $K$ is not properly coloured by our initial assignment. Further, let $k$ be the number of monochromatic edges between the vertices $\{v_{1, t}, v_{2, t}, \dots, v_{m, t}\}$. We partition $K$ into the sets $K_A$ and $K_B$ where $K_A = \{ v_i \mid \psi(v_1) = \psi(v_i)$ and $(v_{1, t}, v_{i, t} \notin E_{\alpha}) \} \cup \{ v_i \mid \psi(v_1) \neq \psi(v_i)$ and $(v_{1, t}, v_{i, t} \in E_{\alpha}) \}$, and $K_B = K \setminus K_A$. Observe that the set of vertices $K_A$ are properly coloured with respect to $v_{1, t}$, and thus each other, while the set of vertices are improperly coloured relative to $v_{1, t}$, and thus properly coloured with respect to each other.
    Therefore, to convert this into a proper colouring, we can ``flip'' the colours of the vertices in either $K_A$ or $K_B$, setting $\psi_t(v_i)$ to $0$ if it was originally $1$, or $1$ if it was originally $0$, for each $v_i$ in the flipped set.

    To determine the corresponding cost of this new colouring, we now show that $\min(\vert K_A\vert, \vert K_B\vert) \leq k$. First, assume that there exists a pair of vertices $v_{1, t}, v_{2, t} \in K_A$ such that $(v_{1, t}, v_{2, t}) \in E_\alpha$. Then, for every vertex $v_i \in K_B$, there must be at least one monochromatic edge to either $v_{1, t}$ or $v_{2, t}$ in the original partition of $V_{\alpha}$. Thus, in this case $\vert K_B \vert \leq k$. Otherwise, if no such pair exists, then $\{v_{j, t} \in K \mid (v_{1, t}, v_{j, t}) \in E_{\alpha}\} \subseteq K_B$ and, by extension, there exists some vertex $v_{i, t} \in K_B$ such that, $\forall v_{j, t} \in K_A$, $(v_{i, t}, v_{j, t})$ is monochromatic, and thus $\vert K_A \vert \leq k$, completing the claim.
    
    Note that for each vertex $v_i$ where the colour is changed when flipping the colours in either $K_A$ or $K_B$ we add at most two changes to $v_i$ in $\Psi$, corresponding to potentially changing $v_i$ between snapshots $t - 1$ and $t$, and $t$ and $t + 1$. Thus, recolouring this component adds a cost of at most $2k$ and therefore the final colouring has a cost of at most $2C$.

\vspace{-12.1pt}
\qed\end{proof}

\begin{theorem}
	\label{thm:approximating_budget}
	Given an always-bipartite temporal graph $\mathcal{G}$ with a lifetime $T$ and $n$ vertices we can find a temporal sequence $2$-colouring $\Psi = \psi_1, \psi_2, \dots, \psi_T$ such that $\Cost(\Psi)$ that is at most a factor of $O\left(\sqrt{\log{(nT)}}\right)$ greater than the minimum budget of any such colouring in $\tilde{O}((nT)^3)$ time.
\end{theorem}

\begin{proof}
	Let $\Psi'$ be some minimum cost temporal 2-colouring of $\mathcal{G}$ and observe that, per Lemmas \ref{lem:colouring_to_auxilary_graph} and \ref{lem:auxilary_graph_to_bipartite}, $2M \geq \Cost(\Psi') \geq M$, where $M$ minimum number of monochromatic edges in $\alpha(\mathcal{G})$. By the results of Agarwal et al. \cite{agarwal2005log}, and Arora et al. \cite{arora2007combinatorial}, we can find a partition of $\alpha(\mathcal{G})$ with $M'$ monochromatic edges, where $M'$ is a factor of at most $O\left(\sqrt{\log(nT)}\right)$ more than $M$. From Lemma \ref{lem:auxilary_graph_to_bipartite}, we can convert this assignment to a temporal 2-colouring $\Psi$ of $\mathcal{G}$ that requires a budget of at most $2 M'$. Thus, $ 2M' \geq \Cost(\Psi) \geq \Cost(\Psi') \geq M$, hence $\Cost(\Psi)$ is a factor of at most $O\left(\sqrt{\log(nT)}\right)$ greater than $\Cost(\Psi')$, giving the bound on the approximation factor. To get the time complexity, observe that we can find, by Arora et al. \cite{arora2005fast}, a soloution to any \MUC instance in $\tilde{O}(m^3)$, where $m$ is the number of vertices in the input instance, in our case $n T$, giving a time complexity of $\tilde{O}((nT)^3)$.

\vspace{-12.1pt}
\qed\end{proof}

\section{Parameterised Algorithm for Bipartite Colourings}

Finally, we provide an algorithm for finding the optimal $2$-temporal sequence colouring parameterised by the maximum number $k$ of connected components in any given snapshot.
Our algorithm works in a dynamic manner, computing the $2^k \times T$ sized table, $B$, indexed by the set of potential $2$-colourings of $\mathcal{G}$ and the set of timestamps, $[T]$. Formally, let $\mathcal{G} = G_1, G_2, \dots, G_T$ be a temporal graph such that each snapshot contains at most $k$ connected components. Observe that, for a bipartite graph with at most $k$ components, $\phi_C(G) \leq 2^k$, corresponding to the two unique ways of colouring each component. We use the following observation to aid in the construction of the table $B$ and for determining the time complexity.

\begin{observation}
    \label{obs:getting_all_colourings}
    The set of all bipartite colourings of a graph $G = (V, E)$, $\phi_C(G)$ with at most $k$ components can be output in $O(\vert E \vert 2^k)$ time.
\end{observation}

\paragraph*{Algorithm.}
Let $B$ be a table of size $2^k \times T$ such that $B[\psi, t]$ is the minimum budget needed to have a graph with the colouring $\psi$ in snapshot $G_t$, for any $\psi \in \phi_C(G_t)$ and $t \in [T]$.
As a base case, we set $B[\psi, 1]$ to $0$, for every $\psi \in \phi(G_1)$.
For $t \in [2, T]$, the value of $B[\psi, t]$ is determined from the set of colourings of $G_{t - 1}$ by looking for the colouring $\psi'$ of $\mathcal{G}$ minimising $B[\psi', t- 1] + \Cost((\psi', \psi))$, formally:

$$B[\psi, t] = \min_{\psi' \in \phi_C(G_{t - 1})} B[\psi', t - 1] + \Cost(\psi', \psi).$$

We provide pseudocode for computing the table $B$ in Algorithm \ref{alg:bipartite_colour}, and for converting the table $B$ into a temporal sequence colouring in Algorithm \ref{alg:get_colour}.

\begin{lemma}
    \label{lem:table_B}
    The value of $B[\psi, t]$ can be computed in $O(n2^k)$ time, assuming the value of $B[\psi', t -1], \forall \psi' \in \phi_C(G_{t - 1})$ has been precomputed.
\end{lemma}

\begin{proof}
    Observe that the minimum cost of any temporal sequence colouring $\Psi = \psi_1, \psi_2, \dots, \psi_t$ such that $\psi_t = \psi$ for the given colouring $\psi_t$ can be determined by finding the colouring $\psi'$ of $G_{t - 1}$ such that there exists a temporal colour $\Psi' = \psi_1', \psi_2', \dots, \psi_{t - 1}'$ such that $\psi_{t - 1}' = \psi'$ and $B[\psi', t- 1] + \Cost(\psi', \psi)$ is minimal amongst all colourings of $G_{t - 1}$. As there are at most $2^k$ colourings of $G_{t - 1}$, and we can compute $\Cost(\psi', \psi)$ in $O(n)$ time, we can find $\min_{\psi' \in \phi_C(G_{t - 1})} B[\psi', t - 1] + \Cost(\psi', \psi)$ from $B$ in $O(n 2^k)$, we get the claim.
    
\vspace{-12.1pt}
\qed\end{proof}

\begin{algorithm}
    \caption{Bipartite Colouring Cost Table Algorithm}
    \label{alg:bipartite_colour}
    \begin{algorithmic}
        \Procedure{ColourCostTable}{$\mathcal{G} = (G_1, G_2, \dots, G_T)$}
            \State $k \gets \max_{t \in T} \vert \Components(G_t)\vert$
            \State $B \gets$ Empty $2^k \times T$ Table
            \State $\Psi \gets $ Empty Temporal Sequence Colouring
            \State $B[1, \psi] \gets 0$, $\forall \psi \in \phi_C(G_1)$
            \For{$t \in [2, T]$}
                \For{$\psi \in \phi_C(G_t)$} \Comment{Iterate over the set of all colourings of the graph $G_t$.}
                    \State $B[\psi, t] \gets \infty$ \Comment{Add a place holder value for $B[\psi, t]$.}
                    \For{$\psi' \in \mathbf\Psi(G_{t - 1})$} \Comment{Iterate over the set of all colourings of $G_{t - 1}$.}
                        \If{$B[\psi', t - 1] + \Cost(\psi, \psi') < B[\psi, t]$}
                            \State $B[\psi, t] \gets B[\psi', t - 1] + \Cost(\psi, \psi') $
                        \EndIf
                    \EndFor
                \EndFor
            \EndFor
            \State \textbf{return} $B$
        \EndProcedure
    \end{algorithmic}
\end{algorithm}




\begin{theorem}
    \label{thm:fpt}
    The minimum cost of any temporal sequence $2$-colouring on the temporal graph $\mathcal{G} = G_1, G_2, \dots, G_T$ where no snapshot contains more than $k$ components can be computed in $O(T \vert E \vert 2^k + n T 2^{2k})$ time.
\end{theorem}

\begin{proof}
From Lemma \ref{lem:table_B} we can compute each entry in $B$ in $O(n2^k)$ time, provided we only compute $B[\psi, t]$ once $B[\psi', t- 1]$ has been computed for all $\psi' \in \phi_C(G_{t - 1})$. As there are $O(T 2^k)$ entries in $B$, and we can output the set of all entries in $O(T\vert E \vert 2^k)$, we get a total complexity of $O(T\vert E \vert 2^k + n T 2^{2k})$ for computing the full table $B$. Once the table $B$ has been computed, we can determine the minimum cost colouring of $\mathcal{G}$ in $O(2^k)$ by finding $\min_{\psi \in \phi_C(G_T)} B[\psi, T]$, concluding the proof.

\vspace{-12.1pt}\qed\end{proof}

Note that Theorem \ref{thm:fpt} thus gives an $O(T \vert E \vert)$ time algorithm when the number of components is constant in each round, most notably when the graph is always connected. Finally, we make a brief observation on the complexity of determining a minimum-cost temporal 2-colouring from the table $B$.

\begin{proposition}
    \label{prop:getting_the_colouring_from_b}
    We can determine a temporal 2-colouring of the temporal graph $\mathcal{G} = G_1, G_2, \dots, G_T$ in $O(
T\vert E \vert 2^k + n T 2^{2k})$
\end{proposition}

\begin{proof}
    Let $\Psi = \psi_1, \psi_2, \dots, \psi_T$ be the temporal sequence colouring, and let $B$ be the table as computed in Theorem \ref{thm:fpt}. We work from the colouring $\psi_T$ such that $B[\psi_T, T] = \min_{\psi' \in \phi_C}B[\psi', T]$. We determine the value of $\psi_t$, for every $t \in [T - 1]$, from the value of $\psi_{t + 1}$ by finding a colouring $\psi$ such that $B[\psi, t] + \Cost(\psi, \psi_{t + 1}) = B[\psi_{t + 1}, t + 1]$, which can be done in $O(n2^k)$ time in a brute force manner, giving the bound.
    
\vspace{-12.1pt}\qed\end{proof}

\begin{algorithm}
    \caption{Bipartite Colouring Algorithm}
    \label{alg:get_colour}
    \begin{algorithmic}
        \Procedure{Colouring}{$\mathcal{G} = (G_1, G_2, \dots, G_T)$}
            \State $B \gets$ \textsc{ColourCostTable}$(\mathcal{G})$
            \State $\Psi = (\psi_1, \psi_2, \dots,\psi_T) \gets (\emptyset, \emptyset, \dots, \emptyset)$
            \For{$\psi \in \mathcal{\Psi}(G_T)$}
                \If{$\psi_T = \emptyset$ or $B[\psi, T] < B[\psi_T, T]$}
                    \State $\psi_T \gets \psi$
                \EndIf
            \EndFor
            \For{$t \in T - 1, T - 2, \dots, 1$}
                \For{$\psi \in \mathcal{\Psi}(G_t)$}
                    \If{$C[\psi, t] + \Cost(\psi, \psi_{t + 1}) = B[\psi_{t + 1}, t + 1]$}
                        \State $\psi_t \gets \psi$
                        \State \textbf{break}
                    \EndIf
                \EndFor
            \EndFor
            \State \textbf{return} $\Psi$
        \EndProcedure
    \end{algorithmic}
\end{algorithm}

\section{Conclusion}

In this paper, we have introduced and studied the problem of maintaining a temporal sequence colouring, with a focus on two colourings of bipartite graphs. In doing so, we have provided a strong analysis of the complexity landscape of this problem, including hardness results, a fixed parameter tractability result for the number of components, an $O\left(\sqrt{\log(nT)}\right)$-factor approximation algorithm, and a hardness of constant factor approximation.

There are two natural directions in which to continue this work. The first is to extend the algorithmic results, most naturally the approximation result, to colourings with a larger palette size. It seems likely that a polynomial-time algorithm can be found with a similar approximation factor for maintaining degree + 1 colourings in general, where degree here refers to the maximum degree of any vertex in any snapshot. On one hand, as with bipartite graphs, these can be determined in linear time relative to the number of edges, avoiding the natural problems due to the hardness of the colouring problem in general. On the other hand, unlike with bipartite graphs, the same freedom may preclude the same techniques used here for our approximation results.

The second direction is to improve the approximation bound. While a constant factor approximation bound has been ruled out unless the unique games conjecture fails, it is still possible that an improvement can be found in general, or for some particular classes of graphs. In particular, as our construction requires the use of long paths, the question of approximation remains open for temporal graphs in which each timestep is a matching.

\bibliography{bib}

\end{document}